\documentclass{article}

\usepackage{esvect}
\usepackage{hhline}

\usepackage{courier}
\newif\ifcompileappendix
\compileappendixtrue

\usepackage{esvect}
\usepackage{hhline}

\usepackage{amsfonts,amsmath,amsthm,amssymb}
\usepackage{paralist}
\usepackage{multirow}
\usepackage{comment}
\usepackage{xspace}
\usepackage[noend,ruled]{algorithm2e}
\usepackage{colortbl}
\usepackage{tikz}
\usetikzlibrary{decorations,arrows,petri,topaths,backgrounds,shapes,positioning,fit,calc,decorations.pathreplacing,patterns,intersections}
\usepackage[noend,ruled]{algorithm2e} 
\usepackage{multirow}

\usepackage[numbers,sectionbib]{natbib}
\usepackage{subcaption}

\usepackage{authblk}

\usepackage{esvect}
\usepackage{mathtools}

\usepackage{hhline}
\usepackage{booktabs}
\usepackage{xcolor}

\newlength{\additionaltextwidth}
\usepackage[colorlinks,citecolor=green!60!black,linkcolor=red!60!black,pagebackref]{hyperref}

\usepackage{url}
\usepackage{latexsym} 
\usepackage{dsfont}

\newtheorem{example}{Example}
\newtheorem{proposition}{Proposition}
\newtheorem{theorem}{Theorem}
\newtheorem{lemma}{Lemma}
\newtheorem{corollary}{Corollary}

\newtheorem{observation}{Observation}
\newtheorem{remark}{Remark}

\usepackage{multirow}
\usepackage{multicol}
\usepackage{xspace}
\usepackage{booktabs}
\usepackage{blkarray}

\usepackage[textsize=scriptsize]{todonotes}

\usepackage{footnote}
\makesavenoteenv{tabular}
\makesavenoteenv{table*}

\usepackage{paralist}

\newcommand{\amax}{\ensuremath{a_{\max}}}

\DeclareMathOperator*{\argmax}{arg\,max}

\usepackage{cleveref}

\crefname{section}{Section}{Sections}
\crefname{subsection}{Subsection}{Subsections}
\crefname{table}{Table}{Tables}
\crefname{figure}{Figure}{Figures}
\crefname{algorithm}{Algorithm}{Algorithms}
\crefname{theorem}{Theorem}{Theorems}
\crefname{definition}{Definition}{Definitions}
\crefname{corollary}{Corollary}{Corollaries}
\crefname{proposition}{Proposition}{Propositions}
\crefname{observation}{Observation}{Observations}
\crefname{lemma}{Lemma}{Lemmas}
\crefname{example}{Example}{Examples}
\crefname{reduction}{Reduction}{Reductions}
\crefname{algorithm}{Algorithm}{Algorithms}
\crefname{appendix}{Appendix}{Appendices}

\newcommand{\p}{{\mathsf{P}}}
\newcommand{\np}{{\mathsf{NP}}}
\newcommand{\coNP}{{\mathsf{coNP}}}
\newcommand{\fpt}{{\mathsf{FPT}}}
\newcommand{\wtwo}{{\mathsf{W[2]}}}

\newcommand{\SetCover}{\textsc{Set Cover}\xspace}

\newcommand{\SCS}{\ensuremath{\mathcal{F}}\xspace}
\newcommand{\SCU}{\ensuremath{U}\xspace}
\newcommand{\SCk}{\ensuremath{k}\xspace}
\newcommand{\SCm}{\ensuremath{m}\xspace}
\newcommand{\SCn}{\ensuremath{n}\xspace}

\newcommand{\TPG}{\textsc{TPG}\xspace}

\newcommand{\thresholdnumber}{$t$-\textsc{Threshold Covering}\xspace}
\newcommand{\hornsystem}{\textsc{Horn Constraint System}\xspace}

\newcommand{\largeteams}{\ell}

\newcommand{\mytitle}{Teams in Online Scheduling Polls: Game-Theoretic Aspects}

\newcommand{\myparagraph}[1]{\medskip\noindent\textbf{{#1}.}}

\begin{document}

\newcommand{\AND}{}

\title{\mytitle\thanks{R.\ Bredereck is supported by the DFG fellowship BR 5207/2. N. Talmon is supported by an I-CORE ALGO fellowship.
Main work done while R.\ Bredereck and N. Talmon were with TU Berlin.}}

\author{Robert Bredereck$^1$,  Jiehua Chen$^2$, Rolf Niedermeier$^2$,
  Svetlana Obraztsova$^3$, and Nimrod~Talmon$^4$\\
   $^1$University of Oxford, United Kingdom, \texttt{robert.bredereck@cs.ox.ac.uk}\\
  $^2$TU Berlin, Germany, \texttt{\{jiehua.chen, rolf.niedermeier\}@tu-berlin.de}\\
  $^3$I-CORE, Hebrew University of Jerusalem, Israel, \texttt{svetlana.obraztsova@gmail.com}\\
  $^4$Weizmann Institute of Science, Israel, \texttt{nimrodtalmon77@gmail.com}
}

\maketitle

\begin{abstract}
Consider an important meeting to be held in a team-based organization.
Taking availability constraints into account,
an online scheduling poll is being used in order to decide upon the exact time of the meeting.
Decisions are to be taken during the meeting,
therefore each team would like to maximize its relative attendance in the meeting
(i.e., the proportional number of its participating team members).
We introduce a corresponding game,
where each team can declare (in the scheduling poll) a lower total availability,
in order to improve its relative attendance---the pay-off.
We are especially interested in situations where teams can form coalitions.

We provide an efficient algorithm that,
given a coalition, finds an optimal way for each team in a coalition to improve its pay-off.
In contrast,
we show that deciding whether such a coalition exists is NP-hard.
We also study the existence of Nash equilibria:
Finding Nash equilibria for various small sizes of teams and coalitions can be done in polynomial time 
while it is coNP-hard if the coalition size is unbounded.
\end{abstract}

\section{Introduction}\label{section:introduction}

An organization is going to hold a meeting,
where people are to attend. 
Since people come from different places and have availability constraints,
an open online scheduling poll is taken to decide upon the meeting time. 
Each individual can approve or disapproval of each of the suggested time slots.
In order to have the highest possible attendance,
the organization will schedule the meeting at a time slot with the maximum sum of declared availabilities. 
During the meeting, proposals will be discussed and decisions will be made.
Usually, people have different interests in the decision making, e.g. they are from different teams that each want their own proposals to be put through.
We consider people with the same interest as from the same \emph{team} 
and as a result, each team (instead of each individual) may declare (in the scheduling poll) the number of its members that can attend the meeting at each suggested time slot.

For a simple illustration, suppose that three teams, $t_1$, $t_2$, and $t_3$, 
are about to hold a meeting, either at 9am or at 10am. 
Two members from~$t_1$,
one member from~$t_2$,
and three members from~$t_3$
are available at 9am,
while exactly two members of each team are available at 10am.
The availabilities of the teams can be illustrated as an integer matrix:
\begin{align*}
  A\coloneqq 
  \begin{blockarray}{ccc}
    c_1 & c_2 & \\ [3px]
    \begin{block}{(cc)c}
      2 & 2 & \ \ t_1 \\
      1 & 2 & \ \ t_2 \\
      3 & 2 & \ \ t_3  \\
     \end{block}
   \end{blockarray}
\end{align*}

A time slot is a winner if it receives the maximum sum of declared availabilities.
Thus, if the three teams declare their true availabilities,
then both 9am and 10am co-win (since six people in total are available at 9am and 10am each), 
and the meeting will be scheduled at either 9am or 10am.

Now, if a team (i.e. people with the same interest) wants to influence any decision made during the meeting, 
then it will want to send as many of its available team members to the meeting as possible
because this will maximize its \emph{relative power}---the proportion of its own attendees.
For our simple example, if the meeting is to be held at 9am, 
then the relative powers of teams~$t_1$, $t_2$, and $t_3$ are $1/3$, $1/6$, and $1/2$, respectively.
This again means that a sophisticated team will change its availabilities declared in the poll from time to time.
However, each team must not report a number which is higher than its true availability
since it cannot send more members than available. 
Given this constraint, it is interesting to know whether any team can increase its relative attendance by misreporting about its availability.

Aiming to have more power in the meeting,
our teams can declare a different number than their true availabilities,
possibly changing the winning time slot to one where their relative powers are maximized.
For the case where several time slots co-win, however, 
it is not clear which co-winning time slot will be used.
To be on the safe side, 
the teams must maximize the relative power of \emph{each} co-winning time slot.
In other words, our teams are pessimistic and consider their
\emph{pay-off} as the \emph{minimum} over all the relative powers at each co-winning time slot.
In our example, 
this means that the pay-off of team~$t_2$ would be $1/6$,
since this is its relative power at 9am, which is smaller than its relative power, $1/3$, at 10am.
The pay-offs of teams~$t_1$ and $t_3$ are both~$1/3$.
In this case,
team~$t_2$ can be strategic by updating its availability and declare zero availability at 9am;
as a result, the meeting would be held at 10am, where team~$t_2$ has better pay-off with relative number of $1/3$.

We do not allow arbitrary deviations from the real availabilities of teams;
specifically,
we do not allow a team to declare as available a higher number than actually available.
Further,
we do not allow a team to send more team members to the meeting than it declared as available,
because this is often mandated by the circumstances.
As some examples,
we mention that the organizer might need to arrange a meeting room and specify the number of participants in the meeting up-front
(similarly, if the meeting is to be carried in a restaurant, the number of chairs at the table shall be decided beforehand);
or the organizer might need to obtain buses to transport the participants.
Thus, the teams must send exactly the declared number of members to the meeting.
For instance, it is not possible for team~$t_2$ to declare $3$ at 9am since only one of its team members is available.
A formal description of the corresponding game, called \emph{team power game} (\TPG, in short),
and a discussion on our example are given in~\cref{section:preliminaries}.

As already remarked, 
to improve the pay-off, a team may lie about the number of its available members. 
Sometimes, teams can even form a coalition and update their availabilities strategically. 
In our example, after team~$t_2$ misreported its availabilities such that each team receives a pay-off of $1/3$,
teams~$t_1$ and $t_3$ may collaborate: if both teams keep their declared availabilities at 9am
but declare zero availability at 10am (note that team~$t_2$ does not change its updated availabilities), 
then 9am will be the unique winner (with total availability of $5$);
as a result, $t_1$ and $t_3$ receive better pay-offs of $2/5$ and $3/5$, respectively. 
Such a successful deviation from the declared availabilities of the teams in a coalition 
(while keeping the declared availabilities of the teams not in the coalition unchanged)
is called an \emph{improvement step}.

After some teams perform an improvement step,
other teams may also want to update their availabilities to improve.
This iterative process leads to the question of whether
there is a stable situation, where improvement is impossible---a \emph{Nash equilibrium}.
Of course, when searching for equilibria, 
it is natural to ask how hard it is to decide whether an improvement step is possible. 

In this paper,
we are interested in the computational complexity of the following problems:
(1) finding an improvement step (if it exists) for a specific coalition,
(2) finding an improvement step (if it exists) for any coalition, and 
(3) finding a \emph{$t$-strong Nash equilibrium} (if it exists).

\newcommand{\tablefootnote}[1]{\footnote{#1}}

\iftrue
{
\begin{table}[t]
\caption{Complexity results for the team power game. 
  ``Unary'' (resp.\ ``Binary'') means 
  that the input and the strategy profiles are encoded in unary (resp.\ binary).
  Variable~$t$ stands for the number of teams in a coalition,
  while~$a_{\max}$ stands for the maximum true availability. 
  An entry labeled with ``$\p$'' means polynomial-time solvability.
  An entry labeled with ``$\fpt$ for~$k$'' means solvability in $f(k)\cdot |I|^{O(1)}$ time, 
  where $f$ is a function solely depending on $k$ and $|I|$ denotes the size of the given input.
  An entry labeled with ``$\wtwo$-hard for~$k$''
  implies that the corresponding problem is not ``$\fpt$ for $k$'' unless $\wtwo=\fpt$ (this is considered unlikely in parameterized complexity theory).
}
{\center
\begin{tabular}{@{}l@{}ll@{}}
\toprule
(1) & \multicolumn{2}{l}{Finding an improvement step for a given coalition} \\
& \qquad Unary & in $\p$~(Thm.~\ref{thm:unary-p}) \\
& \qquad Binary & in $\fpt$ for $t$~(Thm.~\ref{thm:binary-unbounded})$^1$\\[1ex]
(2) & \multicolumn{2}{l}{Deciding the existence of an improvement step}\\
    & \multicolumn{2}{l}{for any coalition}\\
& \qquad Binary & in $\p$ for constant~$t$ (Cor.~\ref{cor:poly})\\
& \qquad $a_{\max}= 1$ &  $\np$-complete~(Thm.~\ref{thm:cvc}) \\
& \qquad $a_{\max}= 1$ &  $\wtwo$-hard for $t$ (Thm.~\ref{thm:cvc}) \\[1ex]
(3) & \multicolumn{2}{l}{Finding a $1$-strong Nash equilibrium} \\
& \qquad $a_{\max} \le 3$ & in $\p$, always exists (Thm.~\ref{thm:simple_Nash_always_exists})\\
& \qquad $a_{\max} \ge 4$ & {\color{darkgray}open} (Rem.~\ref{rem:1-NE})\\[.8ex]
 & \multicolumn{2}{l}{Finding a $2$-strong Nash equilibrium} \\
& \qquad $a_{\max} = 1$ & in $\p$, always exists (Prop.~\ref{prop:two_Nash_always_exists})\\
& \qquad $a_{\max} \geq 2$ & {\color{darkgray}open}, does not always exist\\[.8ex]
 & \multicolumn{2}{l}{Deciding the existence of a $t$-strong Nash equilibrium} \\
& \qquad $a_{\max} = 2$ & $\coNP$-hard (Thm.~\ref{thm:strongNash-coNP-hard})\\
\bottomrule
\end{tabular}
\par}

$^1$\footnotesize{We conjecture it to be even in $\p$. Strong $\np$-hardness is excluded by Theorem~\ref{thm:unary-p}.}
\label{table:results}
\end{table}
}
\fi


\myparagraph{Main Contributions}
We show that, depending on the size of the coalition
(i.e., the amount of collaboration allowed),
the computational complexity of finding an improvement step for a given coalition 
and deciding whether an improvement step exists for an arbitrary coalition  
ranges from being polynomial-time solvable to being $\np$-hard; 
further,
deciding whether an improvement step exists for any coalition of size at most~$t$ is $\wtwo$-hard
when parameterizing by the coalition size~$t$.
We show that a $1$-strong Nash equilibrium always exists for some special profiles
and we provide a simple polynomial-time algorithm for finding it in these cases.
Finally, we show that deciding whether a $t$-strong Nash equilibrium exists is $\coNP$-hard.
Our results are summarized in~\cref{table:results}.

\myparagraph{Related Work}
Recently,
online scheduling polls such as Doodle~/~Survey Monkey caught the attention of several researches.
\citet{RNBNG13} initiated empirical investigations of scheduling polls and
identified influences of national culture on people's scheduling behavior,
by analyzing actual Doodle polls from 211~countries.
\citet{zou2015strategic} also analyzed actual Doodle polls,
and devised a model to explain their experimental findings.
They observed that people participating in open polls
tend to be more ``cooperative'' and additionally approve time slots that are very popular or unpopular;
this is different to the behavior of people participating in closed polls. 
\citet{OEPR15} formally modeled the behavior observed by \citet{zou2015strategic} as a game,
where approving additional time slots may result in pay-off increase.
While the game introduced by \citet{OEPR15} captures the scenario that each individual player tries to appear to be cooperative,
our team power game models the perspective that each individual team (player) as a whole tries to maximize
its relative power in the meeting, which means that approving more time slots is not necessarily a good strategy.

Quite different in flavor,
\citet{lee2014algorithmic} considered a computational problem 
from the point of view of the poll initiator,
whose goal is to choose the time slots to poll over,
in order to optimize a specific cost function.
Finally,
since scheduling polls might be modeled as approval elections,
we mention 
the vast amount of research done on approval elections in general, e.g.,~\cite{brams1978approval}
and on iterative approval voting in particular, e.g.,~\cite{DORK15,lev2012convergence,MPRJ10}.

\newcommand{\cowinners}{\mathsf{winners}}
\newcommand{\teampower}{\mathsf{team}\text{-}\mathsf{power}}
\newcommand{\payoff}{\mathsf{pay}\text{-}\mathsf{off}}
\newcommand{\sumpayoffA}{\mathsf{s}_A}
\newcommand{\sumpayoff}{s}

\section{Preliminaries}\label{section:preliminaries}

We begin this section by defining the rules of the game which is of interest here.
Then,
we formally define the related computational problems we consider in this paper.
Throughout, given a number~$n\in \mathds{N}$,
by $[n]$ we mean the set~$\{1,2,\ldots, n\}$.

\myparagraph{Rules of the Game}
The game is called the \emph{team power game} (\TPG, in short).
It consists of $n$ players,
the \emph{teams},
$t_1,t_2,\ldots,t_n,$
and $m$ possible \emph{time slots},
$c_1,c_2,\ldots,c_m$. 
Each team $t_i$ is associated with a \emph{true availability vector}
%
$A_i = (a_i^1, a^2_i, \ldots,a_i^m),$ 
%
\noindent where $a_i^j \in \mathds{N}$ is the (true) availability of team~$t_i$ for time slot~$c_j$.
Importantly, each team is only aware of its own availability vector.
During the game,
each team~$t_i$ announces a \emph{declared (availability) vector}
$B_i=(b^1_i,b^2_i,\ldots,b^m_i),$
%
\noindent where $b_i^j\le a_i^j$ is the declared availability of team~$t_i$ for time slot~$c_j$;
using standard game-theoretic terms,
we define the \emph{strategy} of team~$t_i$ to be its declared availability vector~$B_i$.
We use~$A$ and~$B$ to denote the matrices consisting of a row for each team's true and declared availability vectors. 
That is, for $i \in [n]$ and $j \in [m]$, $A \coloneqq (a^j_i), B \coloneqq (b^j_i)$.
Given a declared availability matrix~$B$,
the co-winners of the corresponding scheduling poll, 
denoted as $\cowinners(B)$,
are the time slots with the maximum sum of declared availabilities: 
%
\begin{align*} 
\cowinners(B) \coloneqq \argmax\limits_{c_j \in \{c_1,c_2\ldots,c_m\}} \{\sum_{i \in [n]} b_i^j\}\text{.}
\end{align*}

Before we define the \emph{pay-off} of each team,
we introduce the notion of \emph{relative power}.
The relative power $\teampower(B,t_i,c_j)$ of team $t_i$ at time slot~$c_j$
equals the number of members from $t_i$
who will attend the meeting at time slot $c_j$,
divided by the total number of attendees at this time slot:
\begin{align*}
\teampower(B,t_i,c_j) \coloneqq \frac{b_i^j}{\sum_{k\in [n]} b_k^j}\text{.}
\end{align*}


In order to define the pay-off of each team,
we need to decide how to proceed when several time slots tie as co-winners.
In this paper we consider a \emph{maximin} version of the game,
where ties are broken adversarially.
That is, the pay-off of team $t_i$ is defined to be the minimum,
  over all co-winners,
  of its \emph{relative power}:
%
\begin{align*}
\payoff(B,t_i) \coloneqq \min\limits_{c_j \in \cowinners(B)} \teampower(B,t_i,c_j)\text{.}
\end{align*}


When we refer to an \emph{input} for \TPG,
we mean a true availability matrix~$A\in \mathds{N}^{n\times m}$
where 
each row~$A_i$ represents the true availability of a team~$t_i$ for the $m$ time slots.
When we refer to a \emph{strategy profile} (in short, \emph{strategy}) for input~$A$ 
we mean a declared availability matrix~$B\in \mathds{N}^{n\times m}$ 
where
each row~$B_i$ represents the declared availability vector of team~$t_i$.
%

\myparagraph{Computational Problems Related to the Game}
Given a \emph{coalition}, i.e., a subset of teams,
a deviation of the teams in the coalition from their current strategy profile is an \emph{improvement step}
if,
by this deviation,
each team in the coalition strictly improves its pay-off.
Given a positive integer~$t$,
a \emph{$t$-strong Nash equilibrium} for some input $A$ is a strategy profile $B$
such that no coalition of at most $t$ teams has an improvement step wrt.~$B$.
We are interested in the following computational questions:
%

1. Given an input, a strategy profile, and a coalition of at most $t$ teams,
  does this coalition admit an improvement step compared to the given strategy profile?

2. Given an input, a strategy profile, and a positive integer~$t$,
  is there any coalition of at most $t$ teams which has an improvement step compared to the given strategy profile?

3. Given an input and a positive integer~$t$,
  does a \emph{$t$-strong Nash equilibrium} for this input exist?


We are particularly interested in understanding the dependency of the computational complexity of the above problems
on the number $t$ of teams in a coalition.
Specifically,
we consider
(1) $t$ being a constant (modeling situations where not too many teams are willing to cooperate or where cooperation is costly) and
(2) $t$ being unbounded.


\myparagraph{Illustrating Example}
Consider the input matrix~$A$ given in Section~\ref{section:introduction}, 
which specifies the true availabilities of three teams~$t_1,t_2,t_3$ over two time slots~$c_1,c_2$.
%
If all teams declared their true availabilities,
then both time slots win with total availability $6$.
The pay-offs of the teams~$t_1,t_2,t_3$ are $1/3,1/6,1/3$, respectively.
Team~$t_2$ can improve its pay-off by declaring~$(0,2)$ (i,e., declaring $0$ for $c_1$ and $2$ for $c_2$).
As a result, $c_2$ would become the unique winner with total availability~$6$
and team~$t_2$ would receive a better pay-off: $1/3$.
%
%
\noindent Thus, the profile $B$ for~$A$ where all teams declare their true availabilities
(i.e., where $B = A$)
is not a $1$-strong Nash equilibrium.
Nevertheless, $A$ does admit the following $1$-strong Nash equilibrium:
\begin{align*} 
B'\coloneqq \begin{pmatrix}
    2 & 2 \\
    1 & 2 \\
    3 & 0
  \end{pmatrix}
\end{align*}

The declared availability matrix~$B'$,
however,
is not a $2$-strong Nash equilibrium,
since if team~$t_1$ and~$t_2$ would form a coalition and declare the same availability vector~$(0,2)$,
then $c_2$ would be the unique winner with total availability $4$ 
and both~$t_1$ and $t_2$ would have a better pay-off: $1/2$.

\section{Improvement Steps}\label{section:results_improvement_steps}

We begin with the following lemma, 
which basically says that,
in search for an improvement step, 
a fixed coalition of teams 
needs only to focus on a single time slot.

\begin{lemma}\label{one_is_enough}
  %
  %
  If a coalition has an improvement step wrt.\ a strategy profile~$B$,
  then it also has an improvement step~$E=(e^j_i)$ wrt.\ $B$,
  where there is one time slot~$c_k$ such that 
  each team~$t_i$ in the coalition declares zero availability for all other time slots
  (i.e., $e^{j}_{i} = 0$ holds for each team~$t_i$ in the coalition and each time slot $c_{j}\neq c_k$).
\end{lemma}
The missing proof for \cref{one_is_enough} can be found in \cref{proof:one_is_enough}.

By \cref{one_is_enough},
we know that if a fixed coalition has an improvement step for a strategy profile~$B$,
then it admits an improvement step that involves only one time slot.
Assume that it is time slot $c_k$.
In order to compute an improvement step for the coalition,
we first declare zero availabilities for the teams in the coalition,
for all other time slots.
Then,
we have to declare specific availabilities for the teams in the coalition,
for time slot $c_k$.
This is where the collaboration between the teams comes into play:
even though each team,
in order to improve its pay-off,
might wish to declare as high as possible availability for time slot $c_k$
(i.e., its true availability),
the teams shall collaboratively decide on the declared availabilities,
since a too-high declared availability for one team might make it impossible for another team
(even when declaring the maximum possible amount, i.e., the true availability)
to improve its pay-off.
It turns out that this problem is basically equivalent to the following problem
(which, in our eyes, is interesting also on its own).

\myparagraph{Relation to Horn Constraint Systems} 
%
\iftrue
Using \cref{one_is_enough},
we know that if a coalition of $t$ teams admits an improvement step for 
a strategy profile~$B$,
then it admits an improvement step that involves only one time slot.
Let $c_k$ be such a time slot.
Then, finding an improvement step that involves only time slot~$c_k$ 
reduces to the following number problem.
\fi
given a natural number vector~$(a_1,a_2,\ldots,a_{t})\in \mathds{N}^t$, 
 a rational number vector~$(p_1,p_2,\ldots,p_t)\in \mathds{Q}^{t}$
 with $\sum_{i\in[t]}p_i\le 1$,
 and a natural number~$p\in \mathds{N}$,
 searches for a natural number vector~$(x_1,x_2,\ldots,x_t)$ where for each $i\in [t]$ the following holds:
  \begin{inparaenum}[(1)]
    \item $1\le x_i \le a_{i}$ and
    \item $x_i/(p+\sum_{i\in[t]}x_i) > p_i$.
  \end{inparaenum}
%
%
%

Intuitively,
the vector $(a_1,\ldots, a_t)$ corresponds to the true availabilities of the teams in the coalition in time slot $c_k$,
while the solution vector $(x_1, \ldots, x_t)$ corresponds to the declared availabilities of the teams in the coalition in time slot $c_k$;
accordingly,
the first constraint makes sure that each declared availability is upper-bounded by its true availability.
Further,
the vector $(p_1, \ldots, p_t)$ corresponds to the current pay-offs of the teams in the coalition,
while $p$ corresponds to the sum of the declared availabilities of the teams not in the coalition at time slot $c_k$;
accordingly,
the second constraint makes sure that,
for each team in the coalition,
the new pay-off is strictly higher than its current pay-off.
More formally,
we argue that the coalition~$\{t_1,t_2,\ldots,t_t\}$ has an improvement step compared to strategy profile~$B$,
involving only time slot~$c_k$,
if and only if 
the instance $(A^*, P, p)$
for \thresholdnumber has a solution,
where $A^* \coloneqq (a^{k}_1,\ldots,a^{k}_t)$, $P \coloneqq (\payoff(B,t_1),\ldots,\payoff(B,t_t))$, and $p \coloneqq \sum_{i \in [n]\setminus [t]}b^{k}_i$.


\begin{remark}
Since the values~$p_i$ ($i \in [t]$) are rational numbers,
we can rearrange the second constraint in the description of \thresholdnumber
to obtain an integer linear feasibility problem.
This means that 
\thresholdnumber is a special variant of the so-called \hornsystem problem which,
given a matrix~$U=(u_{i,j})\in \mathds{R}^{n'\times m'}$ with each row having at most one positive element, 
a vector~$b \in \mathds{R}^{n'}$, and a positive integer~$d$, 
decides the existence of an integer vector~$x\in \{0,1,\ldots,d\}^{m'}$ such that $U\cdot x \ge b$;
\hornsystem is weakly $\np$-hard and can be solved in pseudo-polynomial time~\cite{Lagarias1985,Lagarias1985}.
\end{remark}

%
%

Taking a closer look at \thresholdnumber, 
we observe the following: 
if we would know the sum of the variables $(x_1, \ldots, x_t)$,
then we would be able to directly solve our problem by checking every constraint and taking the smallest feasible value
(i.e., given $\sum_{i \in [t]} x_i$, we would set each $x_i$ to be the minimum over all values satisfying 
$x_i / (p + \sum_{i\in[t]}x_i) > p_i$).
This yields a simple polynomial-time algorithm for finding an improvement step for the likely case
where all availabilities are polynomially upper-bounded in the input size;
technically, this means where the input profile~$A$ is encoded in unary.

\begin{theorem}\label{thm:unary-p}
  Consider an input~$A$ and a strategy profile~$B$. 
  Let $\sumpayoff$ be the sum of all entries in~$A$.
  Finding an improvement step (if it exists) for a given coalition is solvable in $O(s^2)$~time.
\end{theorem}
The proof for \cref{thm:unary-p} can be found in \cref{proof:unary-p}.

Indeed,
\thresholdnumber can be reduced to finding the sum $\sum_{i \in [t]} x_i$.
If the input is encoded in binary,
however, 
then this sum might be exponentially large in the number of bits that encode our input,
thus we cannot simply enumerate all possible values.
If the coalition size~$t$ or a certain parameter~$\largeteams$ that measures the number of ``large'' true availabilities
is a constant,
then we still have polynomial-time algorithms for which the degree of the polynomial in the running time does not depend on the specific parameter value.
Specifically, by the famous Lenstra's theorem~(\cite{Len83}, later improved by~\citet{Kan87} and by~\citet{FraTar1987}),  we have the following result. 


\begin{theorem}\label{thm:binary-unbounded}
  Consider an input~$A$ and a strategy profile~$B$.
  Let $L$ be the length of the binary encoding of $A$.
  For each of the following times~$T$, there is a $T$-time algorithm 
  that finds an improvement step, compared to~$B$, for a given coalition of $t$ teams:
  \begin{enumerate}
    \item $T = O(t^{2.5t+o(t)} \cdot L^2)$ and
    \item for each constant value~$c$,
    $T = f(\largeteams_c) \cdot t^2 \cdot L^{c^2+2}$,
  \end{enumerate}
  where $f$ is a computable function and $\largeteams_c \coloneqq \max_{j}|\{i \in [t]\colon a_{i}^{j}> L^{c}\}|$ is the maximum over the numbers of teams~$t_i$ in the coalition that have true availabilities~$a_i^{j}$ with $a_i^{j}>L^c$ for the same time slot~$c_j$. 
\end{theorem}
The proof for \cref{thm:binary-unbounded} can be found in \cref{proof:binary-unbounded}.
 
Using \cref{thm:binary-unbounded},
and checking all $\sum_{i=1}^{t}\binom{n}{i}$ possible coalitions of size at most~$t$,
we obtain the following.

\begin{corollary}\label{cor:poly}
  Given an input and a strategy profile, 
  we can find,
  in polynomial time,
  a coalition of a constant number of teams and,
  for this coalition,
  find an improvement step compared to the given profile.
\end{corollary}


In general,
however,
deciding whether an improvement step exists is
computationally intractable
as the next result shows.
We briefly note that,
under standard complexity assumptions,
a problem being $\wtwo$-hard for parameter~$k$ 
presumably excludes any algorithm with running time~$f(k)\cdot |I|^{O(1)}$,
where $f$ is a computable function depending only on $k$
and $|I|$ is the size of the input.

\begin{theorem}\label{thm:cvc}
  Given an input and a strategy profile, 
  deciding whether there is a coalition of size~$t$ that has an
  improvement step is $\wtwo$-hard wrt.~$t$
  even if all teams are of size one.
  It remains $\np$-complete if there
  is no restriction on the coalition size.
\end{theorem}

\begin{proof}(Sketch).
%
%
To show $\wtwo$-hardness, we provide a parameterized reduction from the \SetCover problem,
which is $\wtwo$-complete wrt.\ the set cover size~$\SCk$~\cite{DF13}:
Given sets $\SCS=\{S_1,\dots,S_{\SCm}\}$ over a universe~$\SCU=\{u_1,\dots,u_{\SCn}\}$ of elements and a positive integer~$\SCk$,
\textsc{Set Cover} asks whether
there is a size-$\SCk$ \emph{set cover}~$\SCS' \subseteq \SCS$, i.e., $|\SCS'|= \SCk$ and $\bigcup_{S_i\in \SCS'}S_i=\SCU$.
The idea of such a parameterized reduction is,
given a \SetCover instance~$(\SCS,\SCU,\SCk)$,
to produce,
in $f(\SCk)\cdot (|\SCS|+|\SCU|)^{O(1)}$ time,
an equivalent instance~$(A,B,t)$ such that $t\le g(\SCk)$,
where $f$ and $g$ are two computable functions.
%
%
Let $(\SCS,\SCU,\SCk)$ denote a \SetCover instance.
For technical reasons, 
we assume that each set cover contains at least three sets.

\noindent \textbf{Time slots.}\quad
For each element~$u_j\in\SCU$, we create one \emph{element slot}~$e_j$.
Let $E:=\{e_1,\dots,e_{\SCn}\}$ denote the set containing all element slots.
We create two special time slots: $\alpha$ (the original winner) and~$\beta$ (the potential new winner).

\noindent \textbf{Teams and true availabilities~$A=(a^j_i)$.}\quad
For each set~$S_i \in \SCS$,
we create a \emph{set team}~$t_i$ that has true availability~$1$ at time slot~$\alpha$, 
at time slot~$\beta$, 
and at each element slot~$e_j$ with~$u_j \in S_i$.
We introduce several dummy teams, as follows.
Intuitively,
the role of these dummy teams is to allow to set specific sums of availabilities for the time slots;
the crucial observation in this respect is that the dummy teams do not have any incentive to change their true availabilities,
therefore we can assume that they do not participate in any coalition.
For each element~$u_j$,
let $\#(u_j)$ denote the number of sets from~$\SCS$ that contain~$u_j$.
For each element slot $e_j$,
we create $(2\SCm - 1 - \#(u_j))$ dummy teams such that each of these dummy teams has availability~$1$ at element slot~$e_j$
and availability~$0$ for all other time slots.
Similarly,
for time slot~$\alpha$, 
we create $\SCm$~additional dummy teams,
each of which has availability~$1$ for time slot~$\alpha$
and availability~$0$ for all other time slots.
For time slot~$\beta$,
we create $2\SCm-1-\SCk$ further dummy teams,
each of which has availability~$1$ for time slot~$\beta$
and availability~$0$ for all other time slots.

\noindent \textbf{Declared availabilities~$B=(b^j_i)$.}\quad
Each dummy team declares availability for the time slot where it is available.
Each set team declares availability for all time slots where it is available except for time slot~$\beta$ where all set teams declare availability~$0$.

We set the size of the coalition $t$ to be $k$.
This completes the reduction which can be computed in polynomial time.
Indeed, it is also a parameterized reduction.
%
%
%
A formal correctness proof as well as the extension
to the case of unrestricted coalition sizes to show the $\np$-hardness result are deferred to \cref{proof:thm-cvc}.
\end{proof}

Taking a closer look at the availability matrix constructed in the proof of \cref{thm:cvc}, 
we observe the following.

\begin{corollary}
Deciding the existence of an improvement step
for any coalition is $\np$-hard, even for very sparse availability matrices, i.e.,  
even if each team has only one team member 
and is truly available at no more than four time slots.
\end{corollary}

\section{Nash equilibria}\label{section:results_nash_equilibria}

We move on to consider the existence of Nash equilibria.
Somewhat surprisingly, it seems that,
a $1$-strong Nash equilibrium always exists.
Unfortunately, we can only prove this 
when the maximum availability $\amax \coloneqq \max_{i \in [n], j \in [m]} a_i^j$
is at most three.
Extending our proof strategy to $\amax\ge4$ seems to require a huge case analysis.

\begin{theorem}\label{thm:simple_Nash_always_exists}
  If the maximum availability $\amax$ is at most three, then
  \TPG always admits a $1$-strong Nash equilibrium.
\end{theorem}

\begin{proof}(Sketch).
 Let $A=(a^j_i)$ be the input profile.
 We begin by characterizing two simple cases for which $1$-strong Nash equilibria always exist.
 
 \medskip\noindent\textbf{Safe single-team slot.} Suppose that a time slot $c_{j}$ exists where only
 one team, $t_{i}$, is available with some availability~$a^*$
 (i.e., $a^{j}_{i}=a^*$),
 all other teams are not available in this time slot
 (i.e., $a^{j}_{i'}=0$ for all $i' \neq i$), and
 no other team, $t_{i'}, i' \neq i$, is available with availability greater than~$a^*$
 at any time slot.
 Then, we obtain a $1$-strong Nash equilibrium~$B=(b^l_i)$ by setting
 $b_i^j \coloneqq a_i^j$,
 and,
 for each $i' \in [n]$ and each $j' \neq j$,
 setting $b_{i'}^{j'} \coloneqq 0$;
 to see why $B$ is a Nash equilibrium,
 notice that the only team (namely $t_i$) i.e. available at time slot~$c_j$ 
 already receives the best possible pay-off (namely $1$)
 and no other team can prevent~$c_j$ from being a co-winner,
 which would be necessary to improve their pay-off (which is $0$).
 We call such time slot~$c_j$ a \emph{safe single-team~slot}.

 \smallskip\noindent\textbf{Safe multiple-team slot.} Suppose that a time slot~$c_j$ exists where multiple teams
 have non-zero true availabilities and no single team is powerful enough to 
 prevent~$c_j$ from co-winning, by declaring zero availability.
 That is, for each team~$t_i$ and each time slot~$c_{j'} \neq c_j$,
 it holds that $a^{j'}_i \leq \sum_{i'\neq i} a^j_{i'}$.
 Again, we obtain a $1$-strong Nash equilibrium~$B=(b^j_i)$ by setting
 $b_l^j \coloneqq a_l^j$ for each team~$t_l$
 and setting $b^{j'}_l \coloneqq 0$
 for each other time slot~$c_j \neq c_{j'}$.
 We call such time slot~$c_j$ a \emph{safe multiple-team slot}.
 For example, the following input profile contains two safe multiple-team slots, namely $c_1$
 and $c_4$:

{
\centering
$
  A\coloneqq \ \ \begin{blockarray}{ccccc}
    c_1 & c_2 & c_3 & c_4 \\
    \begin{block}{(cccc)c}
      1 & 2 & 0 & 0 & \ \ t_1 \\
      2 & 0 & 2 & 0 & \ \ t_2 \\
      1 & 0 & 0 & 1 & \ \ t_3 \\
      0 & 1 & 1 & 3 & \ \ t_4 \\
    \end{block}
  \end{blockarray}
$\par}

We are ready to consider instances $\amax \le 3$.

\smallskip\noindent\textbf{Instances with $\amax = 2$.}
%
 Consider the maximum availability sum~$x$ of all time slots,
 i.e., the maximum column sum of the matrix~$A$.
 Clearly, $x \ge \amax$.
 We proceed by considering the different possible values of $x$.

\smallskip\noindent\textit{Cases with $x = 2$:} 
 If $x$ is two, then there is a time slot where only one
 team is available with availability $\amax=2$.
 Thus,
 there is a safe single-team slot.

\smallskip\noindent\textit{Cases with $x = 3$:} 
 If $x$ is three,
 then,
 since we have $\amax=2$, it follows that either
 (1) there is a safe single-team slot where only one team is available with availability $\amax=2$
 or
 (2) there is a time slot~$c_j$ where a single team~$t_i$ has availability $\amax=2$ and another team~$t_{i'}$ has availability $1$.
 In the first case,
 there is a safe single-team,
 so let us consider the second case.
 To this end,
 let $c_j$ be the time slot such that a single team~$t_i$ has availability $\amax=2$ and another team~$t_{i'}$ has availability $1$.
 Next,
 we show how to construct a $1$-strong Nash equilibrium~$B=(b^j_i)$. 
 First, for each team~$t_i$, set $b_{i}^j\coloneqq a^j_i$ and $b_{i}^{j'}\coloneqq 0, j'\neq j$.
 This makes time slot $c_j$ the unique winner.
 Team~$t_i$ receives pay-off~$2/3$
 and team $t_{i'}$ receives pay-off~$1/3$.
 Second, 
 for each time slot~$c_{j'} \neq c_j$,
 if $a^{j'}_{i'} > 0$, then set~$b^{j'}_{i'}\coloneqq 1$;
 otherwise, find any team~$t_k \neq t_i$ with non-zero availability $a^{j'}_{k} = 1$
 and set~$b^{j'}_{k}\coloneqq 1$.
 In this way, every time slot except~$c_j$ has total availability one (if there is at least one team with non-zero availability for this slot).
 Thus, $c_j$ remains a unique winner and the declared total availabilities of other time slots make it impossible for any team to improve:
   First,
   team~$t_i$ cannot improve because it would receive the same pay-off~$2/3$ for every time slot which it could make a new single winner
   (recall that no safe single-team slot exists).
   Second,
   team~$t_{i'}$ also cannot improve 
   because it cannot create a new single winner at all.
   Last,
   neither of the remaining teams can improve because they cannot prevent $c_j$ from co-winning.
 Hence, we have a $1$-strong Nash equilibrium.

\smallskip\noindent\textit{Cases with $x \geq 4$:} 
 Every time slot~$c_j$ with availability sum~$x$ is a safe multiple-team slot
 since $\forall j'\colon a^{j'}_i \le \amax=2$
 and
 $\forall i\colon \sum_{i'\neq i} a^j_{i'}\ge x-\amax \ge 2$.
\medskip

We defer the proof details for instances with $\amax = 3$ to Appendix~\ref{proof:simple_Nash_always_exists}.
\end{proof}

\begin{remark}\label{rem:1-NE}
We do not know 
any instances without $1$-strong Nash equilibria.
However,
we could not generalize our proof even for instances with $\amax = 4$.
Nevertheless, some general observations from our proof hold for every $\amax$.
In particular, if there is a column with only one entry with~$\amax$ (a special case of a safe single-team slot)
or if the maximum column sum is at least~$2\amax$
(a special case of a safe multiple-team slot), then
a $1$-strong Nash equilibrium exists.
\end{remark}

Since our proof is constructive, we obtain the following.

\begin{corollary}
  If the maximum availability $\amax$ is at most three, then
  a $1$-strong Nash equilibrium for \TPG can be found in polynomial time.
\end{corollary}

The situation where $t \ge 2$ is quite different
already with only two teams.
By a proof similar to the case of $t=1$ and $\amax=2$,
we can show that a $2$-strong Nash equilibrium always
exists for $t = 2$ and $\amax=1$:

\begin{proposition} \label{prop:two_Nash_always_exists}
  If the maximum availability $\amax$ is one, then
  a $2$-strong Nash equilibrium for \TPG always exists and can be found in polynomial time.
\end{proposition}

Complementing~\cref{thm:simple_Nash_always_exists},
we demonstrate that $2$-strong Nash equilibria do not always exist,
even when $\amax = 2$;
to this end,
consider the following example.




\begin{example}\label{example:Nash_not_always_exists}
We provide in the following an example with maximum true availability two,
and show that it does not admit a $1$-strong Nash equilibrium.
Consider the following input for \TPG.

{\centering
  $  A\coloneqq  \begin{blockarray}{ccc}
    c_1 & c_2 & \\ [3px]
    \begin{block}{(cc)c}
      2 & 0 & t_1 \\
      2 & 2 & t_2 \\
      0 & 2 & t_3 \\
    \end{block}
  \end{blockarray}
$\par}

Informally, 
\noindent The main crux of this example is that $t_1$
(or, symmetrically, $t_3$)
can cooperate with $t_2$;
in such cooperation,
$t_2$ can choose whether to be `in favor` of $t_1$ or $t_3$,
by declaring either $b_1^2 = 2$ and $b_2^2 = 0$ (favoring $t_1$),
or $b_1^2 = 0$ and $b_2^2 = 2$ (favoring $t_3$).
Moreover,
$t_1$ or $t_3$ can `reward' $t_2$ by not declaring its true availability,
which is $2$,
but only $1$.
In such  a cooperation,
both $t_2$ and $t_1$ (or $t_2$ and $t_3$) strictly improve their pay-offs.

More formally,
let us consider all possibilities for $t_2$.
By symmetry,
we can assume without loss of generality,
that $t_2$ declares either $(0,1)$, $(0,2)$, $(1,1)$, $(1,2)$, or $(2,2)$
(declaring~$(0,0)$ is not possible in a Nash equilibrium---see the analog $(1,1)$ case)).

To this end,
we denote the declared availability matrix~$B$
simply by the symbolic vector $[b^1_1 b^2_1,b^1_2 b^2_2,b^1_3 b^2_3]$
(e.g.\ the declared availability matrix~$B=A$ is denoted by $[20,22,02]$).
We write $B \to^X B'$ if the coalition~$X$ receives a better pay-off
by changing the declared availabilities from~$B$ to $B'$.

Now, we consider each of the possible strategies that team~$t_2$ may take:
\begin{description}
 \item[$t_2$ declares $(0, 1)$:]\ \\
  $[x0,01,0y] \to^{\{t_3\}} [x0,01,02]$ $(0\le x \le 2 ,0\le y \le 1)$;
  $[x0,01,02] \to^{\{t_2\}} [x0,02,02]$ $(0\le x \le 2)$. 
 \item[$t_2$ declares $(0,2)$:] \ \\
  $[x0,02,0y] \to^{\{t_3\}} [x0,02,02]$ $(0\le x \le 2 ,0\le y \le 1)$;\\
  $[x0,02,02] \to^{\{t_1,t_2\}} [10,20,02]$ $(0\le x \le 2)$.
 \item[$t_2$ declares $(1, 1)$:]\ \\
  $[x0,11,0y] \to^{\{t_2\}} [x0,20,0y]$ $(x>y)$;\\
  $[x0,11,0y] \to^{\{t_2\}} [x0,20,0y]$ $(x=y)$;\\
  $[x0,11,0y] \to^{\{t_2\}} [x0,02,0y]$ $(x<y)$.
 \item[$t_2$ declares $(1,2)$:] \ \\
  $[x0,12,0y] \to^{\{t_3\}} [x0,12,02]$ $(0\le x \le 2,0\le y \le 1)$;\\
  $[x0,12,02] \to^{\{t_2\}} [x0,20,02]$ $(0\le x \le 1)$;\\
  $[20,12,02] \to^{\{t_1,t_2\}} [10,20,02]$.

 \item[$t_2$ declares $(2,2)$:]\ \\
 $[x0,22,0y] \to^{\{t_1\}} [20,22,0y]$ $(0\le x \le 1,0\le y \le 1)$;\\
 $[20,22,02] \to^{\{t_1,t_2\}} [10,20,02]$;\\
 $[20,22,01] \to^{\{t_2\}} [20,02,01]$;\\
 $[10,22,02] \to^{\{t_2\}} [10,20,02]$.
\end{description}

The above case analysis cover (by symmetry) all possible cases,
and thus,
shows that there are no $2$-strong Nash equilibria for the instance~$A$.
\end{example}
Naturally,
there are profiles for which $t$-strong Nash equilibria do exist:
  consider, for example, $A$ being the all-one matrix.
%


Next, we show that if the coalition size is unbounded, 
then finding a Nash equilibrium becomes $\coNP$-hard.

\begin{theorem}\label{thm:strongNash-coNP-hard}
  Deciding whether a Nash equilibrium exists for a given input is $\coNP$-hard.
\end{theorem}

\begin{proof}(Sketch).
We reduce from the complement of the following $\np$-complete problem~\cite{G84}: \textsc{Restricted X3C}, which given sets $\mathcal{F}=\{S_1, \ldots, S_{3n}\}$, each containing exactly $3$ elements from $E = \{e_1, \ldots, e_{3n}\}$ such that (1) $n \geq 2$ and  (2) each element $e_i$ appears in exactly $3$ sets,
asks 
whether there is a size-$n$ \emph{exact cover}~$\mathcal{F'}\subseteq \mathcal{F}$, i.e., $|\SCS'|= n$ and $\bigcup_{S_i\in \SCS'}S_i=\SCU$.


Given an instance~$(\mathcal{F}, E)$ of the complement of \textsc{Restricted X3C} we construct a game.
For each element $e_i$ ($i \in [3n]$) we construct a time slot $e_i$.
We construct one additional time slot, denoted by $\alpha$.
For each set~$S_j$ ($j \in [3n]$) we construct a team $s_j$.
For a team $s_j$,
we set its availability for time slot $e_i$, namely $a_i^j$,
to be $n$
if $e_i \in S_j$,
and otherwise $0$.
We set the availability of all teams to be $1$ in time slot $\alpha$.
We consider $2n$-strong Nash-equilibria; thus, we consider coalitions containing up to $2n$ teams.
This finishes the description of the polynomial-time reduction.
The correctness proof can be found in \cref{proof:strongNash-coNP-hard}.
\end{proof}

\section{Conclusion}\label{section:outlook}

We introduced a game considering power of teams (referred to as \TPG) that is naturally motivated by online scheduling polls
where teams declare and update their availabilities in a dynamic process to increase their relative power. 
Our work leads to several directions for future research.


\smallskip
\noindent \textbf{Tie-breaking rules:}
In this paper the teams are pessimistic,
i.e.,
in case of several co-winners,
the pay-off is defined as the \emph{minimum} of the relative number, over the co-winners.
This corresponds to situations where ties are broken adversarially.
We chose this tie-breaking as a standard and natural one,
and as one which models teams which are pessimistic in nature,
where having too low power in the team might have very bad consequences.
Naturally,
one might study other tie-breaking rules 
such as breaking ties uniformly at random or breaking ties lexicographically;
we mention that most of our results seem to transfer to lexicographic tie-breaking.

\smallskip
\noindent\textbf{More refined availability constraints:}
In the online scheduling polls considered in this paper,
the availability constraints expressed by the participants are dichotomous:
each participant can only declare either ``available'' or ``not available'' at each time slot.  
Sometimes,
the availability constraints of people participating in scheduling polls
are more fine-grained;
  for example,
  a participant might not be sure whether she is available or not for some of the suggested time slots,
  but can only provide a ``maybe available'' answer for these time slots.
Correspondingly,
it is interesting to study \TPG
when we allow participants to express more refined availability constraints,
maybe even allowing them to fully rank the time slots according to their constraints.

\smallskip
\noindent\textbf{Nash modification problem:}
Taking the point of view of the poll convener
(who desires to reach a Nash equilibrium),
we suggest to study the following problem:
  given an input for \TPG,
  what is the minimum number of time slots that shall be removed
  so that the modified input will have a Nash equilibrium?

\newcommand{\bibremark}[1]{}

\ifcompileappendix

\clearpage
\appendix

\section{Missing Proofs}\label{apx}
Below we provide proofs missing from the main text.

\subsection{Proof of Lemma~\ref{one_is_enough}}
\label{proof:one_is_enough}
\begin{proof}
Suppose that,
for a certain coalition,
$D=(d^j_i)$ is an improvement step for $B=(b^{j}_i)$.
Let $c_{j_1}$ and $c_{j_2}$ be two time slots such that at least one team~$t_i$ in the coalition has $d^{j_1}_{i} \neq 0$
and at least one team~$t_{i'}$ (possibly different) in the coalition has $d^{j_2}_{i'}\neq 0$.

We distinguish between two cases.
If 
$c_{j_1}\notin \cowinners(D)$ (or~$c_{j_2} \notin \cowinners(D)$),
then the pay-offs of all teams remain the same even if each team~$t_i$ in the coalition with non-zero declared availability~$d^{j_1}_{i}$ (or~$d^{j_2}_{i}$)
for~$c_{j_1}$ (or~$c_{j_2}$) will declare zero instead.
Otherwise,
both~$c_{j_1}$ and $c_{j_2}$ are co-winners in~$D$. 
Fix any team $t_i$ in the coalition.
Since strategy~$D$ is an improvement step compared to $B$,
it follows that 
the relative powers of $t_i$ in $c_{j_1}$ and in $c_{j_2}$ for strategy~$D$
are both strictly larger than the pay-off of $t_i$ for strategy~$B$. 
Thus,
we can change the improvement step~$D$ to make all teams in the coalition declare zero availability, say, in $c_{j_2}$,
and we would still obtain an improvement step.
By repeating the above reasoning,
we fix a desired improvement step where all teams in the coalition have non-zero availability only in the same single time slot.
\end{proof}

\subsection{Proof of Theorem~\ref{thm:unary-p}}
\label{proof:unary-p}
\begin{proof}
Following~\cref{one_is_enough},
we begin by guessing the unique time slot $c_k$ for which the teams in the coalition will declare non-zero availability.

Further,
we guess the value of~$w\coloneqq\sum_{i=1}^{t} x_i$.
Crucially,
this value is upper-bounded by $\sum_{i=1}^{t}a_{i}^{j} \leq s n$;
recall that $s$ is the sum of all true availabilities and $n$ is the number of teams.
Then,
for each team~$t_i$ in the coalition,
we compute the minimum value needed for $t_i$ to get strictly better pay-off than $\payoff(B,t_i)$,
i.e., \linebreak
$\min\{x\le a_{i}^{j} \mid x / (p + w) > \payoff(B,t_i)\}$,
where $p$ is the sum of the declared availabilities of the teams not in the coalition;
we pick these values as the $x_i$'s.
If as a result we obtain $\sum_{i=1}^{t}x_i \le w$,
then we return $(x_1, \ldots, x_t)$.
Otherwise,
we proceed to the next guess.
If all guesses fail, then we reject.
The running time is $O(s^2)$.
\end{proof}

\subsection{Proof of Theorem~\ref{thm:binary-unbounded}}
\label{proof:binary-unbounded}
\begin{proof}
Following~\cref{one_is_enough},
we begin by guessing the unique time slot $c_k$ for which the teams in the coalition will declare non-zero availability.

After guessing the unique time slot $c_k$ for which the teams in the coalition would declare non-zero availabilities,
we can run the integer linear program (ILP, in short) specified for the \thresholdnumber problem. 
By the famous result of \citet{Len83} (later improved by~\citet{Kan87} and by~\citet{FraTar1987}), 
we know that 
an ILP with $\rho$~variables and $z$ input bits can be solved in $O(\rho^{2.5\rho+o(\rho)}\cdot z)$ time.
Since the ILP specified for \thresholdnumber 
has $t$ variables and can be represented in $O(|A'|\cdot |B'|)$ bits,
where $|A'|$ and $|B'|$ denote the numbers of binary bits needed to encode the true and the declared availabilities at time slot~$c_j$, respectively,
we obtain an algorithm with running time $O(t^{2.5t+o(t)} \cdot L^2)$.

As for the second running time,
after guessing the unique time slot $c_k$,
we additionally guess the sum~$w'$ of declared availabilities of the teams in the coalition whose true availabilities are upper-bounded by $L^c$;
we call these teams \emph{small teams}.
Then, we modify the ILP specified for \thresholdnumber 
to search for the declared availabilities of the remaining teams.
Using the declared availabilities for the remaining teams,
we can calculate the declared availabilities for the small teams
as described in the first algorithm.
Again, by Lenstra's result, we can solve the ILP in $g(\largeteams_c)\cdot |A_j'|\cdot |B_j'|$~time, 
where $|A_j'|$ and $|B_j'|$ denote the numbers of binary bits needed to encode the true and declared availabilities at time slot~$c_k$, respectively.
Searching for the declared availabilities for the small teams can be done in $O(t^2\cdot L^{c^2})$ time.
In total, we obtain an algorithm with running time~$f'(\largeteams_c)\cdot t^2\cdot L^{c^2+2}$. 
%
%
%
\end{proof}

\subsection{Proof of Theorem~\ref{thm:cvc} \\ (correctness, unrestricted coalition size)}\label{proof:thm-cvc}

\begin{proof}
Let us state the total availabilities of the various time slots.
Each element slot has total availability~$2\SCm-1$,
time slot~$\beta$ has total availability~$2\SCm-1-\SCk$,
and time slot~$\alpha$ has total availability~$2\SCm$.
Indeed,
time slot~$\alpha$ is a unique winner.
Each set team receives a pay-off of~$1/(2\SCm)$
and each dummy team receives pay-off zero.

Next, we prove the correctness of the reduction.
For the ``if'' case,
assume that there is a size-$k$ set cover $\SCS'\subseteq\SCS$.
We show that the coalition that corresponds to the set cover~$\SCS'$ 
(recall that $t = k$) can improve by making time slot~$\beta$ the unique winner.
To this end,
each set team~$t_i$ with $S_i \in \SCS'$ changes its declared availability for time slot~$\alpha$ and for all element slots to~$0$
and changes its declared availability for time slot~$\beta$ to~$1$.
As a result,
time slot~$\alpha$ has total availability~$2\SCm-k$,
each element slot has total availability at most~$2\SCm-2$
(since the coalition corresponds to a set cover),
and time slot~$\beta$ has total availability~$2\SCm-1$.
Thus,
time slot~$\beta$ is the unique winner and
each set team receives a pay-off of~$1/(2\SCm-1)$;
this is a strict improvement for all teams in the coalition.

For the ``only if'' case,
assume that there is a coalition of $t$~teams that can improve their pay-offs by changing their declared availabilities.
We observe that time slot~$\alpha$ cannot be a (co-)winner since if it was,
then either no team would improve or the pay-off of at least one coalition member would be zero.
Now, we show that the subfamily~$\SCS'$ corresponds to the coalition is a set cover of size~$t=k$.
For this, we distinguish between two cases, depending on whether time slot~$\beta$ is a unique winner or not.

First,
suppose that the coalition makes time slot~$\beta$ become the new unique winner.
Then,
after changing the coalition's declared availabilities,
time slot~$\beta$ has total availability of at most~$2\SCm-1$.
It follows that each element time slot has total availability of at most~$2\SCm-2$.
This implies that $\SCS'$,
which corresponds to the set teams of the coalition,
forms a set cover,
since otherwise there would be one element slot for which the coalition cannot decrease its total availability to be at most~$2\SCm-2$.

Next,
suppose that $\beta$ is not a unique winner which means that there is some subset~$E'\subseteq E$ of element slots within the co-winner set
(including the case that some element slot~$e_j$ is a unique winner).
Note that
if the coalition contains only one team~$t^*$,
then time slot~$\alpha$ would still be a (co-)winner which is not possible by the arguments above. 
Thus,
let us assume that the coalition has at least two teams that make all element slots from~$E'$ become (co-)winners.
All coalition members must still declare availability~$1$ for all element slots from~$E'$,
since otherwise the pay-off of some coalition members would decrease to zero. 
Furthermore,
for each element slot~$e_j\notin E'$ that is not a co-winner,
there is at least one coalition member that changes its declared availability from~$1$ to~$0$,
since otherwise $e_j$~would be a co-winner.
Hence, subfamily~$\SCS'$ is a set cover (recall that it corresponds to the teams in the coalition):
  each element corresponding to an element slot from~$E'$
  is covered by all sets from $\SCS'$
  and
  each element corresponding to an element slot from~$E \setminus E'$
  is covered by at least one set from~$\SCS'$.

By the analysis of the two cases above,
the existence of a coalition that can improve always implies the existence of a set cover.
This completes the proof for the case that the size of the coalition is at most~$t$.

\medskip\noindent\textbf{Unrestricted coalition size.}\quad
Our problem is in $\np$ as we can check in polynomial-time whether 
a specific strategy is an improvement step for a given coalition.
If there is no restriction on the coalition size,
then NP-hardness does not immediately transfer from the above construction,
but it can be obtained by a slight modification as follows.

The above proof would almost work through:
  a size-$k$ set cover in the original instance still implies a coalition of size~$t$ that can improve by making~$\beta$ be the new unique winner.
  Furthermore,
  by the same arguments as above,
  it still follows that every coalition that can improve still consists of set teams which correspond to a set cover.
We cannot,
however,
exclude the existence of a coalition which corresponds to a set cover of size larger than~$t$.
To fix this,
we assume that each element in our \SetCover instance occurs in at most three sets.
(This variant remains NP-hard since \textsc{Vertex Cover} remains NP-hard even if each vertex has degree at most three~\cite{GJS76};
we did not assume this in the $\wtwo$-hardness proof since this variant is not $\wtwo$-hard.)
We further assume, without loss of generality, that~$k\ge3$.

The reasoning for the restricted coalition size case still holds,
so it only remains to show that no coalition of size larger than~$t$ can improve.
Assume,
towards a contradiction,
that there is a coalition of size at least~$t+1$ that can improve.
First,
notice again that time slot~$\alpha$ cannot be a (co\nobreakdash-)winner.
Second,
assume that there is some subset~$E'\subseteq E$ of element slots within the set of co-winners
(including the case that some element slot~$e_j$ is a single winner).
All coalition members must still declare availability~$1$ for all element slots from~$E'$,
since otherwise the pay-off of some team member would decrease to zero.
However,
this is not possible since each element occurs in at most three sets and $t+1=k+1>3$.
Third,
assume that the coalition makes time slot~$\beta$ be the new unique winner.
Then,
after changing the coalition's declared availabilities,
time slot~$\beta$ has total availability of at least~$2\SCm$
since otherwise some coalition members would decrease their pay-off.
However,
this also means that no coalition member improves its pay-off,
which is at most the same as the original pay-off~$1/(2\SCm)$.
This completes the proof for the case of unrestricted coalition size.
\end{proof}

\subsection{Proof of Theorem~\ref{thm:simple_Nash_always_exists} (instances with $\amax \neq 2$)}\label{proof:simple_Nash_always_exists}
\begin{proof}

\medskip\noindent\textbf{Instances with $\amax = 1$.}
 Considering the above two cases,
 the proof is relatively simple for inputs with $\amax=1$:
 either there is a column~$j$ in~$A$ where only one team has availability~$1$,
 implying that time slot~$c_j$ is a safe single-team slot,
 or every time slot~$c_j$ is a safe multiple-team slot since
 $\forall j'\colon a^{j'}_i \le \amax=1$ and
 $\forall i\colon \sum_{i'\neq i} a^j_{i'}\ge 1$.

 \medskip\noindent\textbf{Instances with $\amax = 3$.}
 We consider the case where the maximum true availability~$\amax$ is three.
 Again, let $x$ be the maximum sum of availabilities over all time slots,
 and notice that $x \geq \amax$.
 If $x$ is at least six, then every time slot~$c_j$ with
 availability sum~$x$ is a safe multiple-team slot, because
 $\forall j'\colon a^{j'}_i \le \amax=3$ and
 $\forall i\colon \sum_{i'\neq i} a^j_{i'}\ge x-\amax \ge 3$.
 If $x$ is three, then there is a time slot where only one
 team is available with availability $\amax=3$, i.e.,
 there is a safe single-team slot.

 Now, assume that $x$ is four or five.
 We distinguish between four cases and implicitly assume
 that the $k$th case does not hold in the $(k+1)$th case.
 First, there is a safe single-team slot
 (which implies a $1$-strong Nash equilibrium by our observation).
 Note that this case includes time slots where only one team
 is available with availability $\amax=3$
 as well as the situation that there is a time slot~$c_j$
 where only one team~$t_i$ is available with availability~$2$
 and~$t_i$ is the only team with availability~$2$ for every time slot.
 
 Second, there is a time slot~$c_j$ where one team~$t_i$
 has availability $\amax=3$ and another team~$t_{i'}$
 has availability one while every remaining team has availability zero.
 Analogously to the case with $\amax=2$ and $x=3$,
 first, for each team~$t_\ell$, we set $b_\ell^j\coloneqq a_\ell^j$ and $b_\ell^{j'}\coloneqq 0$ ($j'\neq j$
 to make time slot $c_j$ a single winner.
 Team $t_i$ receives pay-off~$3/4$
 and team $t_{i'}$ receives pay-off~$1/4$.
 Now, for each team~$t_\ell$ and for each time slot~$c_{j'}\neq c_j$, 
 we modify the declared availabilities~$b_{\ell}^{j'}$
 as follows.
 If there is some $\ell \notin \{i,i'\}$ with $a^{j'}_\ell=3$,
 then set~$b^{j'}_\ell\coloneqq a^{j'}_\ell$.
 Otherwise, if $a^{j'}_{i'} > 0$, then set~$b^{j'}_{i'}\coloneqq a^{j'}_{i'}$
 and if $a^{j'}_{i'} = 0$, then set~$b^{j'}_{\ell}\coloneqq a^{j'}_{\ell}$
 for the first position~$\ell \notin \{i,i'\}$ with $a^{j'}_{\ell}>0$.
 This does not prevent $c_j$ from being the unique winner but
 makes it impossible for the teams to improve their pay-offs:
 team~$t_i$ cannot improve, because it would receive
 at most the same pay-off~$3/4$ for every time slot which
 it could make a new single winner.
 Note that, since we are not in the first case,
  it follows that there is no time slot where only team~$t_i$
  is available with availability~$\amax=3$.
  Furthermore, if there is a time slot~$c_{j^*}$ where team~$t_i$
  is available with availability~$2$ and no other team is available,
  then there is some time slot $c_{j''}$
  with $a^{j''}_{i''}=3$ ($i''\neq i$) and, hence,
  $b^{j''}_{i''}=3$.
  (This slot~$c_{j''}$ must exist since otherwise we would be
   in the first case and have a safe single-team slot~$c_{j^*}$.)
  Thus, team~$t_i$ cannot make $c_{j^*}$
  become a new single winner.
 Team $t_{i'}$ also cannot improve, because it cannot
 create a new single winner at all.
 Hence, we have a $1$-strong Nash equilibrium.

 Third, there is a time slot~$c_j$ where one team~$t_i$
 has availability $\amax=3$ and another team~$t_{i'}$
 has availability two, implying that the remaining teams are not available at slot~$c_j$.
 Similarly to the previous case,
 for each team~$t_\ell$, we first set $b^j_\ell\coloneqq a_\ell^j$ and $b^{j'}_\ell\coloneqq 0$ ($j'\neq j$)
 to make time slot $c_j$ a single winner.
 Team $t_i$ receives pay-off~$3/5$
 and team $t_{i'}$ receives pay-off~$2/5$.
 Now, for each team~$t_\ell\neq t_i$ and for each time slot~$c_{j'}\neq c_j$ with $a^{j'}_i=\amax=3$, 
 we modify its declared availability~$b^{j'}_{\ell}$ 
 such that the total declared availability sum of the teams other than $t_i$ is always two at time slot~$c_{j'}$.
\begin{itemize}
  \item If $a^{j'}_{i'} \ge 2$, then set~$b^{j'}_{i'}\coloneqq a^{j'}_{i'}$.
  \item If $a^{j'}_{i'} = 1$, then set~$b^{j'}_{i'}\coloneqq a^{j'}_{i'}$
 and let $t_\ell\notin \{t_i,t_{i'}\}$ be the first team with $0<a^{j'}_{\ell}<\amax$
 and set~$b^{j'}_{\ell}\coloneqq a^{j'}_{\ell}$.
 Note that such a team~$t_\ell$ must exist since we are not in the second case.
 \item If $a^{j'}_{i'} = 0$, then
 set~$b^{j'}_{\ell}\coloneqq a^{j'}_{\ell}$
 for the first team~$t_\ell \notin \{t_i,t_{i'}\}$ with $a^{j'}_{\ell}>1$
 or for the first two teams~$t_\ell \notin \{t_i,t_i'\}$ with $a^{j'}_{\ell}=1$.
 Again, such teams~$t_\ell$ exist(s) since we are not in the second case.
 This does not prevent $c_j$ from being the unique winner but
 makes it impossible for the teams to improve their pay-offs:
 team~$t_i$ cannot improve, because it would receive
 at most the same pay-off~$3/5$ for every time slot which
 it could make a new (co-)winner.
 Each of the remaining teams (including $t_{i'}$) also cannot improve, because it cannot
 create a new single winner at all (note that $x=5$).
\end{itemize}
Hence, we have a $1$-strong Nash equilibrium.
 
 Fourth, there is a time slot~$c_j$ where one team~$t_i$
 has availability~$\amax=3$ and there are another two teams, $t_{i'}$ and~$t_{i''}$, both with availability one.
 Similarly to the previous case,
 for each team~$t_\ell$, 
 we first set $b^j_\ell\coloneqq a^j_\ell$ and $b^{j'}_\ell\coloneqq 0$ ($j'\neq j$)
 to make time slot~$c_j$ a single winner.
 Team $t_i$ receives pay-off~$3/5$
 and teams~$t_{i'}$ and~$t_{i''}$  both receive pay-off~$1/5$.
 Then, for each time slot~$c_{j'} \neq c_j$ such that $a^{j'}_i=\amax=3$,
 we modify the declared availabilities~$b^{j'}_\ell$ of some teams~$\ell$: 
 \begin{itemize}
   \item If $a^{j'}_{i'} > 0$ and $a^{j'}_{i''} > 0$,
   then set~$b^{j'}_{i'}\coloneqq 1$ and set~$b^{j'}_{i''}\coloneqq 1$.
   \item If $a^{j'}_{i'} = 0$ and $a^{j'}_{i''} > 1$,
   then set~$b^{j'}_{i''}\coloneqq a^{j'}_{i''}$. 
   \item If $a^{j'}_{i'} = 0$ and $a^{j'}_{i''} = 1$,
   then set~$b^{j'}_{i''}\coloneqq a^{j'}_{i''}$,
   let $t_\ell$ be the first team with $0<a^{j'}_{\ell}<\amax$, and set~$b^{j'}_{\ell}\coloneqq a^{j'}_{\ell}$.
   Note that such a team~$t_\ell$ must exist since we are not in the second case.
   \item If $a^{j'}_{i'} = 0$ and $a^{j'}_{i''} = 0$, then 
   let $t_\ell, t_{\ell'}$ be the first two teams with $a^{j'}_{\ell} = a^{j'}_{\ell'}=1$
   and set~$b^{j'}_{\ell}\coloneqq a^{j'}_{\ell}$.
   Note that such a team~$t_\ell$ must exist since we are not in the first three cases.
   The cases with $a^{j'}_{i''} = 0$ follow analogously.
   The new declared availabilities still make $c_j$ a unique winner but
   make it impossible for the teams to improve their pay-offs:
   team~$t_i$ cannot improve, because it would receive
   at most the same pay-off~$3/5$ for every time slot which
   it could make a new (co-)winner.
   Neither of the remaining teams (including~$t_{i'}$ and team~$t_{i''}$) 
   can improve, because
   they cannot create a new single winner at all.
 \end{itemize}
 Hence, we have a $1$-strong Nash equilibrium.
\end{proof}

\subsection{Proof of Theorem~\ref{thm:strongNash-coNP-hard} (correctness)}\label{proof:strongNash-coNP-hard}
We will use the following observation.

\begin{observation}\label{truth_for_winning_slot}
  Let $B$ be a Nash-equilibrium for some input $A$.
  Then,
  for each time slot $c_j \in \cowinners(B)$,
  all teams declare their true availabilities.
\end{observation}

\begin{proof}
Assume, towards a contradiction, that the declared availability of some team $t_i$ in time slot $c_j$
is strictly less than its true availability.
Then,
this team can improve its payoff by declaring its true availability.
\end{proof}

Now we are ready to proof the correctness of the construction for Theorem~\ref{thm:strongNash-coNP-hard}.

\begin{proof}(of Theorem~\ref{thm:strongNash-coNP-hard}).
Now, we show that a $2n$-Nash-equilibrium exists if and only if an exact cover does not exist.

%
\medskip\noindent We show the ``only if'' part
by showing that existence of an exact cover implies non-existence of a Nash-equilibrium.
To this end, we assume that there is an exact cover and prove that a Nash-equilibrium does not exist.

Let $\mathcal{F}'$ be an exact cover. Consider the remaining sets $\bar{\mathcal{F}'} = \{S_1, \ldots, S_{3n}\} \setminus \mathcal{F}'$.
It holds that $|\mathcal{F}'| = n$ and $|\bar{\mathcal{F}'}| = 2n$.
Notice further that,
while $\mathcal{F}'$ covers each element exactly once,
$\bar{\mathcal{F}'}$ covers each element exactly twice.

We further assume,
towards a contraction, that a Nash-equilibrium does exist;
let us denote it by $B$.
In what follows, we consider several cases concerning the structure of profile~$B$, which is a Nash-equilibrium,
and show an improvement step for each of these cases, contradicting the assumption that $B$ is indeed a Nash-equilibrium.

Consider the set of winning time slots of $B$.
First we consider the case where $\alpha \in B$.

\smallskip\noindent\textbf{Case 1.}
In this case, $\alpha \in \cowinners(B)$.
By Observation~\ref{truth_for_winning_slot} we have that the declared availabilities of all teams in $\alpha$ is $1$;
thus, the sum of declared availabilities for all winning slots is $3n$.
Below,
we consider three different sub-cases.

\smallskip\noindent\textbf{Case 1a.}
In this case, $\alpha$ is the only winning time slot, i.e., $\alpha = \cowinners(B)$.
Consider any time slot, say $e_1$, and consider the three teams who are available in it,
say $s_i$, $s_j$, and $s_k$.
We claim that these teams can improve their pay-off, as follows:
  if $s_i$, $s_j$, and $s_k$, will declare availability of $0$ for all time slots except for $e_1$
  for which they will declare availability of $n$,
  then their pay-off would increase from $\frac{1}{3n}$ to~$\frac{1}{3}$.

\smallskip\noindent\textbf{Case 1b.}
In this case, there are at most three other time-slots, besides $\alpha$, with sum of declared availabilities being $3n$,
and there exists a team $s_i$ such that $a^i_j = n$ for each $j$ such that $E_j \in \cowinners(B)$.
We claim that $s_i$ can improve its pay-off, as follows:
  if $s_i$ will declare availability of $0$ for time slot $\alpha$, then its pay-off would increase from $\frac{1}{3n}$ to $\frac{1}{3}$.

\smallskip\noindent\textbf{Case 1c.}
In this case, $\alpha$ is indeed a winning slot, but neither Case 1a nor Case 1b hold.
It follows that,
for every team~$s_j$ there must exist $E_i \in \cowinners(B)$ such that $a^j_i = 0$.
Thus, $\payoff(s_j, B) = 0$ for each team $s_j$.
Recall that $\bar{\mathcal{F}'}$ contains those sets which are not part of the exact cover for the instance of \textsc{Restricted X3C}.
We claim that the set of teams corresponding to the sets in $\bar{\mathcal{F}'}$ can improve their pay-off as follows:
  if each such team will declare availability of~$0$ for all time slots $E_i$ and availability of~$1$ for $\alpha$,
  then their pay-offs would increase to $\frac{1}{3n}$.
This follows since $\bar{\mathcal{F}'}$ covers each element twice,
thus,
by deviating as described above,
the sum of declared availabilities of all time slots $E_i$ would decrease to being at most $n$,
while the sum of availabilities of $\alpha$ would be at least $2n$, making $\alpha$ a unique winner.

\medskip
Next we consider the case where $\alpha$ is not winning.

\medskip\noindent\textbf{Case 2.}
In this case, $\alpha \notin \cowinners(B)$.
By Observation~\ref{truth_for_winning_slot} we have that the declared availabilities in the winning time slots of all teams is $n$,
thus, the sum of declared availabilities for all winning time slots is $3n$.
Below,
we consider two sub-cases.

\smallskip\noindent\textbf{Case 2a.}
In this case, there exists a team $s_i$ such that $a^i_j = n$ for each $j$ such that $E_j \in \cowinners(B)$. 
Since in the instance of \textsc{Restricted X3C} we have that no two sets are the same,
it follows that there are four other teams, besides $s_i$, denoted by $s_j$, $s_f$, $s_p$, and $s_q$,
and two different time slots, denoted by $e_x$ and by $e_y$,
such that it the following conditions hold:

\begin{enumerate}
\item $a^i_x = a^j_x = a^f_x = n$;
\item $a^i_y = a^p_y = a^q_y = n$;
\item $j \notin \{p, q\}$;
\item $E_y \in \cowinners(B).$
\end{enumerate}

Notice that $\payoff(s_j,B) = 0$, $\payoff(s_i,B) = \frac{1}{3}$, and $\payoff(s_f, B) \leq \frac{1}{3}$.
We claim that $s_i, s_j, s_f$ can improve their pay-offs as follows:
  if $s_i, s_f$ will declare availabilities of $n$ for time slot $E_x$ and $0$ for all other slots 
  while $s_j$ will declare availability of $n - 1$ for time slot $E_x$ and $0$ for all other slots,
  then $E_x$ would become the unique winning time slot;
  as a result, the pay-offs of $s_i, s_f$ would increase to $\frac{n}{3n-1}$ while the pay-off of $s_j$ would increase to $\frac{n - 1}{3n-1}$.

\smallskip\noindent\textbf{Case 2b.}
In this case, no team exists i.e. available in all winning time slots.
In this case, the improvement step described in Case 1c is an improvement step, and the reasoning is the same, thus omitted.

\medskip\noindent For the ``if'' part, we have to that 
if no exact cover exists, then there is a Nash equilibrium. 
To this end, 
we will describe a profile $B$ and would argue that, if indeed an exact cover does not exist, then $B$ is a Nash-equilibrium.

The profile~$B$ is as follows.
Each team~$s_j$ declares availability of $n$ for all time slots $e_i$ and availability of $0$ for $\alpha$,
making each time slot a co-winner.
Thus, the pay-off of each team is $0$.
Let us assume, towards a contradiction, that $B$ is not a $2n$-Nash-equilibrium,
i.e., that an improvement step with respect to $B$, involving at most $2n$ teams, exists,
denoted by $B'$.

First of all, we observe that no time slot $e_i$ is a winning time slot (i.e., $e_i \notin \cowinners(B')$ for all $e_i$) because 
otherwise 
after the improvement step~$B'$,
at least $3n - 3$ teams will have pay-off $0$
(this follows since, in the instance of \textsc{Restricted X3C},
each element is included in exactly three sets).
However,
no coalition of at most $3$ teams can decrease the sum of availabilities of all other time slots as to make $e_i$ the unique winner;
thus, we conclude that, after the improvement step, $\alpha$ shall be the unique winning slot.

The maximum sum of availabilities which time slot $\alpha$ might get after an improvement step involving at most $2n$ teams is $2n$.
Therefore, to have that such a deviation is a profitable deviation,
i.e., an improvement step,
the sum of declared availabilities of all time slots $e_i$ has to decrease by at least $n + 1$.
This could happen only if, for each time slot $e_i$,
at least two teams with non-zero availabilities would decrease their declared availabilities.
Thus, an exact cover must exist.
\end{proof}

\fi


\begin{thebibliography}{}

\bibitem[\protect\citeauthoryear{Brams and Fishburn}{1978}]{brams1978approval}
Brams, S.~J., and Fishburn, P.~C.
\newblock 1978.
\newblock Approval voting.
\newblock {\em American Political Science Review} 72(03):831--847.

\bibitem[\protect\citeauthoryear{Dery \bgroup et al\mbox.\egroup
  }{2015}]{DORK15}
Dery, L.~N.; Obraztsova, S.; Rabinovich, Z.; and Kalech, M.
\newblock 2015.
\newblock Lie on the fly: Iterative voting center with manipulative voters.
\newblock In {\em Proceedings of IJCAI-2015},  2033--2039.

\bibitem[\protect\citeauthoryear{Downey and Fellows}{2013}]{DF13}
Downey, R.~G., and Fellows, M.~R.
\newblock 2013.
\newblock {\em Fundamentals of Parameterized Complexity}.
\newblock Springer.

\bibitem[\protect\citeauthoryear{Frank and Tardos}{1987}]{FraTar1987}
Frank, A., and Tardos, {\'E}.
\newblock 1987.
\newblock An application of simultaneous {D}iophantine approximation in
  combinatorial optimization.
\newblock {\em Combinatorica} 7(1):49--65.

\bibitem[\protect\citeauthoryear{Garey, Johnson, and Stockmeyer}{1976}]{GJS76}
Garey, M.~R.; Johnson, D.~S.; and Stockmeyer, L.~J.
\newblock 1976.
\newblock Some simplified {NP}-complete graph problems.
\newblock {\em Theoretical Computer Science} 1(3):237--267.

\bibitem[\protect\citeauthoryear{Gonzalez}{1984}]{G84}
Gonzalez, T.~F.
\newblock 1984.
\newblock Clustering to minimize the maximum intercluster distance.
\newblock {\em Theoretical Computer Science} 38:293--306.

\bibitem[\protect\citeauthoryear{Kannan}{1987}]{Kan87}
Kannan, R.
\newblock 1987.
\newblock Minkowski's convex body theorem and integer programming.
\newblock {\em \bibremark{No string.}Mathematics of Operations
  Research\bibremark{No publisher.}} 12:415--440.

\bibitem[\protect\citeauthoryear{Lagarias}{1985}]{Lagarias1985}
Lagarias, J.~C.
\newblock 1985.
\newblock The computational complexity of simultaneous diophantine
  approximation problems.
\newblock {\em SIAM Journal on Computing} 14(1):196--209.

\bibitem[\protect\citeauthoryear{Lee}{2014}]{lee2014algorithmic}
Lee, H.
\newblock 2014.
\newblock Algorithmic and game-theoretic approaches to group scheduling.
\newblock In {\em Proceedings of AAMAS-2014},  1709--1710.

\bibitem[\protect\citeauthoryear{Lenstra}{1983}]{Len83}
Lenstra, H.~W.
\newblock 1983.
\newblock Integer programming with a fixed number of variables.
\newblock {\em \bibremark{No string.}Mathematics of Operations
  Research\bibremark{No publisher.}} 8:538--548.

\bibitem[\protect\citeauthoryear{Lev and
  Rosenschein}{2012}]{lev2012convergence}
Lev, O., and Rosenschein, J.~S.
\newblock 2012.
\newblock Convergence of iterative voting.
\newblock In {\em Proceedings of AAMAS-2012},  611--618.

\bibitem[\protect\citeauthoryear{Meir \bgroup et al\mbox.\egroup
  }{2010}]{MPRJ10}
Meir, R.; Polukarov, M.; Rosenschein, J.~S.; and Jennings, N.~R.
\newblock 2010.
\newblock Convergence to equilibria in plurality voting.
\newblock In {\em Proceedings of AAAI-2010},  823--828.

\bibitem[\protect\citeauthoryear{Obraztsova \bgroup et al\mbox.\egroup
  }{2015}]{OEPR15}
Obraztsova, S.; Elkind, E.; Polukarov, M.; and Rabinovich, Z.
\newblock 2015.
\newblock Doodle poll games.
\newblock In {\em Proceedings of the 1st IJCAI Workshop on AGT}.

\bibitem[\protect\citeauthoryear{Reinecke \bgroup et al\mbox.\egroup
  }{2013}]{RNBNG13}
Reinecke, K.; Nguyen, M.~K.; Bernstein, A.; N{\"{a}}f, M.; and Gajos, K.~Z.
\newblock 2013.
\newblock Doodle around the world: {O}nline scheduling behavior reflects
  cultural differences in time perception and group decision-making.
\newblock In {\em Proceedings of CSCW-2013},  45--54.

\bibitem[\protect\citeauthoryear{Zou, Meir, and
  Parkes}{2015}]{zou2015strategic}
Zou, J.~Y.; Meir, R.; and Parkes, D.~C.
\newblock 2015.
\newblock Strategic voting behavior in {D}oodle polls.
\newblock In {\em Proceedings of CSCW-2015},  464--472.

\end{thebibliography}
\end{document}